\documentclass[12pt]{article}

\usepackage{amsmath,amsthm,amssymb}
\usepackage{cleveref}
\usepackage{natbib}
\usepackage{multirow}
\usepackage[margin=2cm]{geometry}

\usepackage{tikz}
\usetikzlibrary{arrows.meta}

\newcommand{\smalldot}[3]{
\filldraw[black] (#1) circle (0.5pt) node[anchor=#2]{#3};
}

\title{Computing Approximate Roots of Monotone Functions}
\author{
Alexandros Hollender\footnote{All Souls College, University of Oxford},
Chester Lawrence,
Erel Segal-Halevi\footnote{Ariel University}
}

\newtheorem{theorem}{Theorem}[section]
\newtheorem*{theorem*}{Theorem}
\newtheorem{lemma}[theorem]{Lemma}
\newtheorem{observation}[theorem]{Observation}
\newtheorem{proposition}[theorem]{Proposition}
\newtheorem*{proposition*}{Proposition}
\newtheorem{corollary}[theorem]{Corollary}

\newtheorem*{claim*}{Claim}
\newtheorem{claim}[theorem]{Claim}

\theoremstyle{definition}
\newtheorem{remark}[theorem]{Remark}
\newtheorem{definition}[theorem]{Definition}
\newtheorem{assumption}[theorem]{Technical Assumption}
\newtheorem{question}{Open Question}
\newtheorem*{notation}{Notation}

\newcommand{\ba}{\mathbf{a}}
\newcommand{\bb}{\mathbf{b}}

\newcommand{\bx}{\mathbf{x}}
\newcommand{\by}{\mathbf{y}}

\newcommand{\f}{\mathbf{f}}
\newcommand{\g}{\mathbf{g}}
\newcommand{\h}{\mathbf{h}}

\newcommand{\rr}{\mathbb{R}}
\newcommand{\rb}{\mathbb{R}^2}

\newcommand{\rd}{\mathbb{R}^d}
\newcommand{\rda}{\mathbb{R}^{d+1}}

\newcommand{\sign}{\{-1,0,1\}}

\renewcommand{\epsilon}{\varepsilon}

\newcommand{\er}[1]{\textcolor{blue}{#1}}
\newcommand{\erel}[1]{\er{(Erel says: #1)}}
\newcommand{\alex}[1]{\er{(Alex says: #1)}}
\begin{document}

\maketitle

\begin{abstract}
Given a function $f: [a,b] \to R$, if $f(a) < 0$ and $f(b) >  0$ and $f$ is continuous, the Intermediate Value Theorem implies that $f$ has a root in $[a,b]$. Moreover, given a value-oracle for $f$, an approximate root of $f$ can be computed using the bisection method, and the number of required evaluations is polynomial in the number of accuracy digits.

The goal of this note is to identify conditions under which this polynomiality result extends to a multi-dimensional function that satisfies the conditions of Miranda's theorem --- the natural multi-dimensional extension of the Intermediate Value Theorem.

In general, finding an approximate root might require an exponential number of evaluations even for a two-dimensional function. We show that, if $f$ is two-dimensional and satisfies a single monotonicity condition, then the number of required evaluations is polynomial in the accuracy. 

For any fixed dimension $d$, if $f$ is a $d$-dimensional function that satisfies all $d^2-d$ ``ex-diagonal'' monotonicity conditions (that is, component $i$ of $f$ is monotonically decreasing with respect to variable $j$ for all $i\neq j$),    then the number of required evaluations is polynomial in the accuracy. But if $f$ satisfies only $d^2-d-2$ ex-diagonal conditions, then the number of required evaluations may be exponential in the accuracy.

As an example application, we show that computing approximate roots of monotone functions can be used for approximate envy-free cake-cutting.
\end{abstract}

\section{Introduction}
\subsection{Approximate roots}

Given a function $\f: \rd\to \rd$, a point $\bx\in\rd$ is called a \emph{root} of $\f$ if  $|\f(\bx)|=0$, where $|\cdot |$ denotes the maximum norm; it is called an \emph{$\epsilon$-root} of $\f$ (for some real $\epsilon> 0$) if $|\f(\bx)|\leq \epsilon$.  We assume throughout the paper that $\f$ is accessible only through \emph{evaluation queries}, that is: given an input $\bx$, we can compute $\f(\bx)$ in unit time. How many evaluation queries are needed to compute an $\epsilon$-root of $\f$?

Of course, without further information on $\f$, an $\epsilon$-root might not exist at all. In practice, we often have some information on the value of $\f$ on the boundary of some subset of $\rd$, as well as on the rate in which $\f$ changes as a function of $\bx$. In particular, suppose $d=1$, and consider a one-dimensional function $f$ defined on some interval $[a,b]\subseteq \rr$, and satisfying the following properties:
\begin{enumerate}
    \item $f$ switches its sign along the interval, that is, $f(a)\leq 0\leq f(b)$;%
    \item $f$ is Lipschitz-continuous; that is, for some constant $L>0$ and for all $x,x'\in [a,b]$, $|f(x)-f(x')|\leq L\cdot |x-x'|$.
\end{enumerate}
A Lipschitz-continuous function is continuous, so by the intermediate value theorem, $f$ has a root in $[a,b]$. Moreover, an $\epsilon$-root of $f$ can be found using the famous \emph{bisection method}, also known as \emph{binary search}: let $c := (a+b)/2$; if $f(c)\leq 0$ then recursively search in $[c,b]$; otherwise, recursively search in $[a,c]$. 
After $\lceil \log_2((b-a)L/\epsilon)\rceil$ iterations, the length of the search interval is at most $\epsilon/L$. By the Lipschitz condition, each point in the interval is an $\epsilon$-root of $f$.
Treating $a,b,L$ as constants, we can say that the bisection method requires $O(\log(1/\epsilon))$ evaluations.
Equivalently, to find a root with $k$ significant digits, we need $O(k)$ evaluations, so the run-time is linear in the number of accuracy digits.

A natural question of interest is whether this result can be extended to $d\geq 2$ dimensional functions.
Suppose $\f$ is defined in a $d$-box
$[\ba, \bb]\subseteq \rd$, where $\ba < \bb$.%
\footnote{
The vector-comparison operator $\ba < \bb$ means that $a_i<b_i$ for all $i\in[d]$. In all our computational bounds, we think of $\ba$ and $\bb$ as being constants.
}
A natural extension of the sign-switching property is the following.
\begin{definition}
\label{def:switching-cube}
A $d$-dimensional function $\f: [\ba,\bb]\to \rd$ is called \emph{positive-switching} if for every $i\in[d]$, the $i$-th component of $\f$ (denoted $f_i$) switches its sign from negative to positive along the $i$-th axis. That is, for all $\bx\in[\ba,\bb]$,
\begin{align*}
x_i=a_i \implies f_i(\bx)\leq 0;
\\
x_i=b_i \implies f_i(\bx)\geq 0.
\end{align*}
\end{definition}
If $\f$ is continuous and positive-switching, then Miranda's theorem%
\footnote{Also called the Poincaré-Miranda theorem, or the Bolzano-Poincaré-Miranda theorem.}
--- a multi-dimensional extension of the intermediate-value theorem ---
says that $\f$ has a root in $[\ba, \bb]$.
To find an $\epsilon$-root of $\f$,  we can simply partition the $d$-box $[\ba,\bb]$ into a grid of small $d$-boxes, each of which has a side-length of at most $\epsilon/L$. One of these boxes must contain a root, so one of the grid points must be an $\epsilon$-root, and we can find it by evaluating $\f$ on all grid points. However, the number of required evaluations is $\Theta((1/\epsilon)^d)$, which is exponential in the accuracy ($k$ significant digits require $\Theta(2^{kd})$ evaluations).

Is there an algorithm that requires a number of evaluations that is polynomial in the number of digits? --- In general, the answer is no. 
This fact is probably well-known, but we present it here for completeness. 
\begin{proposition}
\label{prop:2d-insufficient}
For every $\epsilon>0, L>0, d\geq 2$,
every algorithm that finds an $\epsilon$-root for every Lipschitz-continuous (with constant $L$) function that is positive-switching on some $d$-cube 
might need $\Omega((L/\epsilon)^{d-1})$ evaluations in the worst case.
\end{proposition}
Even for $d=2$, the lower bound is exponential in the number of accuracy digits.
This raises the question of whether there are natural special cases in which an $\epsilon$-root can be found more efficiently. 
In this note, we focus on a common special case in which $\f$ satisfies \emph{monotonicity} conditions.
\begin{notation}
Given a function $\f: \rd\to \rd$, for any $i,j\in[d]$ and $\bx\in\rd$,
we denote by $f_{i,j,\bx}: \rr\to \rr$ the function $f_i$ considered as a function of the single variable $x_j$, assuming all the other variables $x_1,\ldots,x_{j-1},x_{j+1},\ldots,x_d$ are fixed to their values in $\bx$.

We say that \emph{$f_i$ is monotonically-increasing(decreasing) with $x_j$} if
for all $\bx$, the function $f_{i,j,\bx}$ is weakly monotonically-increasing(decreasing).
\end{notation}

A function $\f:\mathbb{R}^d\to \mathbb{R}^d$ has $d^2$ potential monotonicity conditions, as each $f_i$ can be monotone with respect to every $x_j$ for all $i,j\in[d]$.
We call the conditions that $f_i$ is monotone with $x_i$ \emph{diagonal conditions}, 
and the conditions that $f_i$ is monotone with $x_j$ for $j\neq i$ \emph{ex-diagonal conditions}.
So there are $d$ potential diagonal conditions and $d^2-d$ potential ex-diagonal conditions.

For $d=2$, one monotonicity condition is sufficient for efficient root-finding.
The proof is different depending on whether the condition is diagonal or ex-diagonal.
\begin{theorem}[Diagonal monotonicity condition]
\label{thm:2d-sufficient}
Let $d=2$. Let $\f: [\ba,\bb]\to \rd$ be a positive-switching Lipschitz-continuous function with constant $L$.
If $f_1$ is  monotonically-increasing with $x_1$, 
then for every $\epsilon>0$,
an $\epsilon$-root of $\f$ can be found using $O(\log^2(L/\epsilon))$ evaluations.%
\footnote{
The same holds if $f_1$ is negative-switching and monotonically-decreasing with $x_1$ (replace $f_1$ by $-f_1$). By symmetry, the same holds if $f_2$ is positive(negative)-switching
and monotonically-increasing(decreasing) with $x_2$.
}
\end{theorem}
In fact, we prove a slightly stronger theorem: 
$\f$ does not have to be Lipschitz-continuous or even continuous; it should only satisfy the Lipschitz condition on intervals of length $\epsilon/L$. That is: for any two points $\bx,\bx'\in \rr$, $|\bx-\bx'|\leq \epsilon/L$ implies $|f(\bx)-f(\bx')|\leq \epsilon$.
This may be useful, as in general, it might be hard to prove that $\f$ is continuous.

We also prove a variant of \Cref{thm:2d-sufficient} in which $\f$ is not positive-switching but \emph{sum-switching}, that is, for all $\bx\in[\ba,\bb]$,
\begin{align*}
x_i=a_i \implies f_i(\bx)\leq 0;
\\
x_i=b_i \implies \sum_{j=1}^i f_j(\bx)\geq 0.
\end{align*}

\begin{theorem}[Ex-diagonal monotonicity condition]
\label{thm:2d-sufficient-exdiagonal}
Let $d=2$. Let $\f: [\ba,\bb]\to \rd$ be a positive-switching Lipschitz-continuous function with constant $L$.
If $f_1$ is  monotonically-decreasing with $x_2$, 
then for every $\epsilon>0$,
an $\epsilon$-root of $\f$ can be found using $O(\log^2(L/\epsilon))$ evaluations.%
\footnote{
The same holds if $f_1$ is monotonically-increasing with $x_2$ 
(replace $x_2$ with $-x_2$ and $f_2$ with $-f_2$).
By symmetry, the same holds if $f_2$ is monotonically-increasing or decreasing with $x_1$.
}
\end{theorem}

\Cref{prop:2d-insufficient} and 
Theorems \ref{thm:2d-sufficient} and  \ref{thm:2d-sufficient-exdiagonal} and their variants are proved in  \Cref{sec:roots-2d}.

\subsection{Multi-dimensional functions}
For $d\geq 3$, the situation is not so bright: even $d^2-2$ monotonicity conditions may be insufficient.
Moreover, even the combination of all $d$ diagonal conditions and $d^2-d-2$ ex-diagonal conditions may be insufficient for efficient root-finding.
\begin{proposition}
\label{prop:dd-insufficient}
For every $\epsilon>0, L>0, d\geq 3$,
every algorithm that finds an $\epsilon$-root for every $L$-Lipschitz-continuous positive-switching function that satisfies $d^2-2$ monotonicity conditions (even including all diagonal conditions and all but two ex-diagonal conditions) might need $\Omega(L/\epsilon)$ evaluations in the worst case.
\end{proposition}
On the positive side, for efficient root-finding, it is sufficient that all $d^2-d$ ex-diagonal monotonicity conditions hold (even with no diagonal condition).
\begin{theorem}
\label{thm:dd-sufficient}
Let $d\geq 2$ be a fixed integer. Let $\f: [\ba,\bb]\to \rd$ be a positive-switching Lipschitz-continuous function with constant $L$.
If $f_i$ is  monotonically-decreasing with $x_j$ for every $i,j\in[d]$, $i\neq j$,
then for every $\epsilon>0$,
an $\epsilon$-root of $\f$ can be found using $O(\log^{\lceil(d+1)/2\rceil}(L/\epsilon))$ evaluations.
\end{theorem}

Finally, we show the necessity of the positive-switching conditions. 
Obviously, if one or more of the $f_i$ is not switching, then $\f$ might not have a root at all.
But we can focus on functions that are guaranteed to have an exact root, and ask how many evaluations are required to find an $\epsilon$-root.
Without monotonicity conditions, the problem requires $\Omega(1/\epsilon)$ evaluations even for a 1-dimensional function \citep{sikorski1984optimal}.
But if a one-dimensional function is monotone on $[a,b]$,
then it has a root if and only if it is switching,
so an $\epsilon$-root can be found in time $O(\log(1/\epsilon))$.
We show that he situation is not the same for $d\geq 2$: if even a single switching condition is missing, then even all $d^2$ monotonicity conditions do not make the problem polynomial-time solvable.
\begin{proposition}
\label{prop:switcing-necessary}
For every $\epsilon>0, L>0, d\geq 2$,
every algorithm that finds an $\epsilon$-root for every $L$-Lipschitz-continuous function that has a root and satisfies $d^2$ monotonicity conditions and $d-1$ positive-switching conditions might need $\Omega(L/\epsilon)$ evaluations in the worst case.
\end{proposition}

\Cref{prop:dd-insufficient}, \Cref{thm:dd-sufficient}
and \Cref{prop:switcing-necessary}
 are proved in \Cref{sec:roots-dd}.

\Cref{thm:dd-sufficient} can be adapted to other  monotonicity conditions.
For example, by replacing the variable $x_j$ with $(a_j+b_j)-x_j$ and replacing $f_j$ with $-f_j$, 
the function $f_j$ remains positive-switching with  $x_j$, 
but now for all $i\neq j$, $f_i$ is monotonically-\emph{increasing} with $x_j$,
and $f_j$ is monotonically-\emph{increasing} with $x_i$ (that is, $2d-2$ ex-diagonal conditions are now increasing). 
However, these simple transformations can generate only $2^d$ combinations, out of $2^{d^2-d}$ potential combinations. In particular, the following case is open.

\begin{question}
Let $\f: [\ba,\bb]\to \rd$ be a positive-switching Lipschitz-continuous function 
such that $f_i$ is monotonically-increasing with $x_j$ for all $i\neq j$.
Is it possible to compute an $\epsilon$-root of $\f$ using $O(poly(\log(1/\epsilon)))$ evaluations?
\end{question}

Another case that remains unsolved is when  $d^2-d-1$ ex-diagonal conditions are satisfied.
\begin{question}
Let $\f: [\ba,\bb]\to \rd$ be a positive-switching Lipschitz-continuous function that satisfies some $d^2-d-1$ ex-diagonal monotonicity conditions (and any number of diagonal conditions). Is it possible to compute an $\epsilon$-root of $\f$ using $O(poly(\log(1/\epsilon)))$ evaluations?
\end{question}
Both open questions are for any $d \geq 3$.
See \Cref{tab:summary} for a summary of results and open questions.

\begin{table}
\begin{center}
\begin{tabular}{|c|c|c|c|}
\hline
\textbf{Dimension ($d$)} & \textbf{\# diagonal} & \textbf{\# ex-diagonal}  & \textbf{Polynomial?} \\
\hline
\hline
1 & $\geq 0$ & $\geq 0$ & Yes (bisection) \\
\hline
\hline
\multirow{4}{*}{2} & $0$ & $0$ & No (\Cref{prop:2d-insufficient})   \\
\cline{2-4}
 & $\geq 1$ & $\geq 0$ & Yes (\Cref{thm:2d-sufficient})  \\
\cline{2-4}
 & $\geq 0$ & $\geq 1$ & Yes (\Cref{thm:2d-sufficient-exdiagonal})
\\
\hline
\hline
\multirow{3}{*}{$d\geq 3$} & $\leq d$ & $\leq d^2-d-2$ & No (\Cref{prop:dd-insufficient}) \\
\cline{2-4}
 & $\geq 0$ & $\geq d^2-d$ & Yes (\Cref{thm:dd-sufficient}) \\
\cline{2-4}
 & $0,\ldots,d$ & $d^2-d-1$ & Open \\
\hline
\hline
\end{tabular}
\end{center}
\caption{
\label{tab:summary}
A table summarizing the possibility of finding an approximate root of a positive-switching function, with number of evaluations polynomial in the accuracy. 
The diagonal conditions are that $f_i$ is \emph{increasing} with $x_i$, and the ex-diagonal conditions are that $f_i$ is \emph{decreasing} with $x_j$ for $j\neq i$
(though for $d=2$ the results hold for the opposite conditions too).
}
\end{table}

\subsection{Approximate envy-free cake allocations}
We present an example application of \Cref{thm:2d-sufficient} for the classic problem of \emph{fair cake-cutting}. 
The input to this problem is a cake, usually represented by the interval $[0,1]$, and $n$ agents with different preferences over cake pieces.
An \emph{allocation} is a partition of the cake into $n$ connected pieces (intervals), and an assignment of a single piece per agent. An allocation is called \emph{envy-free} if each agent (weakly) prefers her piece over the other $n-1$ pieces.
An envy-free allocation is guaranteed to exist whenever the agents' preferences are continuous and satisfy the \emph{hungry agents assumption}, that is, they always prefer a non-empty piece to an empty piece \citep{stromquist1980cut,su1999rental}.

This existence result has recently been extended to group-cake-cutting \citep{segal2020cut}: 
given a fixed number of groups $m\leq n$, and fixed group sizes  $k_1,\ldots,k_m$ with $\sum_j k_j = n$, there always exists a partition of the cake into $m$ connected pieces, and an assignment of some $k_j$ agents to each piece $j$, such that every agent prefers her group's piece over the other $m-1$ pieces.

However, there is provably no algorithm that finds an envy-free allocation using finitely-many queries of the agents' preferences, even in the individual-agent setting, whenever there are three or more agents  \citep{stromquist2008envy}.
In practice, assuming agents are indifferent to very small movements of the cut points (e.g., 0.1 mm), it may be sufficient to find an allocation that is \emph{$r$-near envy-free} --- an allocation for which, for each agent $i$, a movement of at most $r$ in some of the cut points would make the allocation envy-free for $i$.
An $r$-near envy-free allocation can be found using $O((1/r)^{n-2})$ queries, and this is tight in the worst case even for $n=3$ \citep{deng2012algorithmic}.

A natural special case is when the agents' preferences are \emph{monotone}, that is, they always prefer a piece to any subset of it. 
Our second contribution is a general reduction from computing an $r$-near envy-free allocation among agents with monotone preferences to computing an $\epsilon$-root. 
To present the reduction, we define a monotonicity property.
\begin{definition}
\label{def:alternating-monotone}
A function $\f: [\ba,\bb]\to \rd$ is called \emph{alternating-monotone} if it satisfies the following $2d-1$ monotonicity conditions:
\begin{itemize}
\item $d$ diagonal conditions: for every $i\in \{1,\ldots,d\}$, $f_i$ is increasing with $x_i$;
\item $d-1$ ex-diagonal conditions: for every $i\in \{2,\ldots,d\}$, $f_i$ is decreasing with $x_{i-1}$.
\end{itemize}
\end{definition}

\begin{theorem}
\label{thm:ef-to-root}
For every integer $d\geq 2$, for every instance of group cake-cutting with $m=d+1$ groups, 
where all agents have monotone preferences,
there exists a function $\f:\rd\to\rd$ with the following properties:
\begin{enumerate}
\item  
One evaluation of $\f$ requires $O(d^2 \cdot n)$ queries of the agents' preferences.
\item  
Every $1/(2\cdot d^2)$-root of $\f$ corresponds to an $r$-near envy-free allocation.
\item  
$\f$ is Lipschitz-continuous with constant $L=n/r$.
\item 
$\f$ is sum-switching in $[0,1]^d$.
\item 
$\f$  is alternating-monotone in $[0,1]^d$.
\end{enumerate}
\end{theorem}
\Cref{thm:ef-to-root} is proved in \Cref{sec:envy-free}.
Combining it with \Cref{thm:2d-sufficient} yields an algorithm for finding a $r$-near $3$-group envy-free allocation using $O(n\cdot \log^2(n/r))$ queries. This algorithm was already found by \citet{igarashi2022envy}, but \Cref{thm:ef-to-root} links this problem to the much more general problem of finding approximate roots of monotone functions, and may pave the way for extending the result in  future work.

\begin{question}
Given $n$ agents with monotone preferences,
$m\geq 4$ group-sizes $k_1,\ldots,k_m$ with $\sum_j k_j=n$,
and approximation factor $r\in(0,1)$, 
it is possible to find an $r$-near-envy-free cake allocation among $m$ groups using $\text{poly}(n, \log(1/r))$ queries?
\end{question}

\section{Related work}
\label{sec:related}

\subsection{Computing approximate roots}
\citet{sikorski1984optimal} proved that, for general $L$-Lipschitz continuous functions (that are neither switching nor monotone), $\Omega((L/\epsilon)^d)$ evaluations may be needed to find an $\epsilon$-root, for any $d\geq 1$.

For switching functions, the bisection method can be used for $d=1$. \citet{mourrain2002complexity} describes several generalizations of this method for $d\geq 2$. One of them is the \emph{characteristic bisection} method. \citet{vrahatis1986rapid} prove that the number of evaluations required for finding an $\epsilon$-root is \emph{at least} $\log_2(1/\epsilon)$. 
Note that, from later surveys \citep{mourrain2002complexity,vrahatis2020generalizations}, it might be understood that $\log_2(1/\epsilon)$ is an upper bound on the number of evaluations; however, in the original paper $\log_2(1/\epsilon)$ is only a lower bound, we could not find any proof of an upper bound. In fact,  \Cref{prop:2d-insufficient} proves a lower bound of $\Omega(1/\epsilon)$.

\citet{adler2016complexity} define a problem called ZeroSurfaceCrossing, in which the input is a two-dimensional function $\f$ given by circuits rather than by a value oracle. They prove that this problem is PPAD-complete, by reduction from a discrete variant of Brouwer's fixed-point theorem.

\subsection{Computing approximate fixed points}
There is a close connection between roots and \emph{fixed points} (see \Cref{sec:fixed-point}).
In particular, Miranda's theorem is closely related to \emph{Brouwer's fixed point theorem}, which says that a continuous function from a convex compact set to itself has a fixed point. 

The first algorithm for computing an  $\epsilon$-approximation of a Brouwer fixed-point of a general function was given by \citet{scarf1967approximation}, using $\text{poly}(1/\epsilon)$ evaluations. Several other algorithms were given later.
\citet{hirsch1989exponential} proved a lower bound $\Omega(L/\epsilon)$ for $d=2$, and a weaker lower bound for any $d\geq 3$.
\citet{chen2005algorithms} proved the tight bound 
$\Theta((L/\epsilon)^{d-1})$ for any $d\geq 3$.

There are two known special classes of functions, in which an $\epsilon$-fixed-point can be computed more efficiently:
\begin{enumerate}
\item \emph{Contractive functions} --- Lipschitz-continuous functions with constant $L\leq 1$; see e.g. \citet{shellman2003recursive}.
\item \emph{Differentiable functions}, where the derivative of $\f$ is known and can be evaluated as well. This allows the use of Newton's method; see e.g. \citet{smale1976convergent,kellogg1976constructive}.
\end{enumerate}
The above papers did not consider monotonicity properties. 
However, \emph{Tarski's fixed-point theorem} is a different fixed-point theorem that explicitly considers monotonicity.
Tarski's theorem says that, for every complete lattice $Z$, every order-preserving function from $Z$ to itself has a fixed point.
The box $[\ba,\bb]$ can be discretized by cutting it in each dimension using lines at most a distance $\epsilon$ apart. The resulting set of vertices, endowed with the component-wise partial-order ($\bx \geq \by$ iff $x_i \geq y_i$ for all $i\in[d]$), forms a lattice. Since this lattice is finite, it is complete. 

\citet{dang2020computations} prove that, for such a component-wise lattice, and for any order-preserving function $\f$, a fixed-point of $\f$ can be found using $O(\log^d (N))$ evaluations of $\f$, where $N$ is the maximum number of points in each dimension.
This was subsequently improved to 
$O(\log^{2\lceil d/3 \rceil} (N))$ by \citet{fearnley2022faster}, and to 
$O(\log^{\lceil (d+1)/2 \rceil} (N))$ by \citet{chen2022improved}.
We use the latter result to prove \Cref{thm:dd-sufficient}.
Note that their result does not imply our \Cref{thm:2d-sufficient}, as this theorem requires only that a single component of $\f$ be monotone in a single component of the input. 

\subsection{Envy-free cake-cutting}
\citet{deng2012algorithmic} studied approximate envy-free cake-cutting when the agents have very general preferences: they only need to satisfy the hungry-agents assumption.
In this general setting, they proved a tight lower bound  for computing an $r$-near envy-free cake-cutting of $\Omega((1/r)^{n-2})$ queries for any $n\geq 3$.

\citet{deng2012algorithmic} also considered a more specific setting, in which agents have \emph{monotone preferences} (they weakly prefer each piece to its subsets). 
For this setting, they presented an algorithm for computing an $r$-near envy-free cake-cutting for three agents in $poly(\log(1/r))$ queries.
\citet{branzei2022query} showed a matching lower bound of $\Omega(\log(1/r))$ even for the case in which each agent's preferences are represented by an \emph{additive non-negative} value-function.
When the agents' preferences are represented by value functions that are $L$-Lipschitz-continuous, $r$-near envy-freeness also implies that the magnitude of envy experienced by each agent is at most $O(L r)$. When the amount of envy experienced by each agent is bounded by some quantity $r$, we say that the division is $r$-envy-free.


Recently, \citet{igarashi2022envy} showed an algorithm that finds an $r$-envy-free allocation using $O(\log^2(1/r))$ queries, and improves over the results of \citet{deng2012algorithmic,branzei2022query} in two respects:
\begin{itemize}
\item It works not only for $3$ individual agents, but also for $3$ \emph{groups}. 
\item It works also for a \emph{two-layered cake} (see their paper for the definition).
\end{itemize}
We provide an alternative proof for their results for 3 groups (for a one-layered cake only).

Very recently, \citet{hollender2023envy} showed an algorithm that finds an $r$-envy-free allocation among $n=4$ individual agents with monotone value-functions, using $\text{poly}(\log(1/r))$ queries.
As far as we know, the case of $n\geq 5$ individual agents, as well as the case of $m\geq 4$ groups of agents, are open even with additive value-functions,

\section{Finding approximate roots: two dimensions}
\label{sec:roots-2d}

\subsection{Hardness without monotonicity}
\label{sec:fixed-point}
The proof of \Cref{prop:2d-insufficient} is by reduction from the problem of computing a \emph{Brouwer fixed point}.

Given a function $\f: \rd\to \rd$, a point $\bx\in\rd$ is called a \emph{fixed-point} of $\f$ if  $\f(\bx)=\bx$; it is called an \emph{$\epsilon$-fixed-point} of $\f$ (for some real $\epsilon> 0$) if $|\f(\bx)-\bx|\leq \epsilon$.

To show the relation between fixed points and roots, we define a notion of function duality:

\newcommand{\dual}[1]{\overline{#1}}
\begin{definition}
The \emph{dual} of a function $\f: \rd\to \rd$ is a function $\dual{\f}: \rd\to \rd$ defined by: $\dual{\f}(\bx) := \bx - \f(\bx)$.
\end{definition}

The following observations follow directly from the definition.
\begin{observation}
\label{obs:fixed}
(a) The dual of the dual of $\f$ equals $\f$: $\dual{\dual{\f}}=\f$.

(b) For all $\epsilon\geq 0$,  $\bx$ is an $\epsilon$-root of $\f$ if and only if it is an $\epsilon$-fixed-point of $\dual{\f}$.

(c) $\f$ is continuous if and only if $\dual{\f}$ is continuous.

(d) If $\f$ is Lipschitz-continuous with constant $L$, then $\dual{\f}$ is 
Lipschitz-continuous with constant at most $L+1$.
\end{observation}

Suppose a continuous function $\f$ is defined on a $d$-box $[\ba, \bb]\subseteq \rd$. If $\f$ maps this $d$-box to itself (that is, $\ba \leq \f(\bx)\leq \bb$ for all $\bx\in[\ba,\bb]$), then $\f$ has a fixed-point in $[\ba,\bb]$; this is a special case of the famous \emph{Brouwer fixed-point theorem}. There is a close relation between Brouwer's fixed-point theorem and Miranda's root theorem, as shown below.
\begin{proposition}
\label{prop:brouwer-to-miranda}
If a function $\f: [\ba, \bb]\subseteq \rd$ maps $[\ba, \bb]$ to itself,
then its dual $\dual{\f}$ is positive-switching in $[\ba, \bb]$ (Definition \ref{def:switching-cube}). 
In other words: if $\f$ satisfies the conditions to Brouwer's theorem, then $\dual{\f}$ satisfies the conditions to Miranda's theorem. 
\end{proposition}
\begin{proof}
Suppose $\f$ maps $[\ba,\bb]$ to itself. 
For any $i\in[d]$, let $\bx$ be any point for which $x_i=a_i$.
Then:
\begin{align*}
\dual{f}_i(\bx) &= x_i - f_i(\bx) = a_i  - f_i(\bx)
\\
&\leq a_i  - a_i && \text{since $\f(\bx)\geq \ba$}
\\
&= 0.
\end{align*}
Similarly, if $x_i=b_i$, then:
\begin{align*}
\dual{f}_i(\bx) &= x_i - f_i(\bx) = b_i  - f_i(\bx)
\\
&\geq b_i  - b_i && \text{since $\f(\bx)\leq \bb$}
\\
&= 0.
\end{align*}
Therefore, by \Cref{def:switching-cube}, the function $\dual{\f}$ is positive-switching in $[\ba,\bb]$.
\end{proof}

Now we can prove a lower bound on the number of evaluations needed to compute an $\epsilon$-root of a function.
\begin{proposition*}[= \Cref{prop:2d-insufficient}]
For every $\epsilon>0, L>0, d\geq 2$,
let $T(L,\epsilon,d)$ be the worst-case number of evaluations required by an algorithm that finds an $\epsilon$-root for every Lipschitz-continuous (with constant $L$) function that is positive-switching on some $d$-cube.
Then $T(L,\epsilon,d) \in \Omega((L/\epsilon)^{d-1})$.
\end{proposition*}
\begin{proof}
The proof uses lower bounds proved by \citet{hirsch1989exponential,chen2005algorithms} for fixed-point computation.
They proved that any algorithm, that finds an $\epsilon$-fixed-point of every Lipschitz-continuous  function $\f$ that maps some $d$-cube to itself, requires $\Omega((\dual{L}/\epsilon)^{d-1})$ evaluations, where $\dual{L}$ is the Lipschitz constant of $\dual{\f}$.

We now show such an algorithm that requires at most  $T(\dual{L},\epsilon,d)$ evaluations.
Let $\f$ be any Lipschitz-continuous function that maps some $d$-cube $[\ba,\bb]$ to itself. 
By \Cref{prop:brouwer-to-miranda}, its dual $\dual{\f}$ is positive-switching. By \Cref{obs:fixed}(d), 
$\dual{\f}$ is Lipschitz-continuous; denote its Lipschitz constant by $\dual{L}$.
By assumption, there is an algorithm that finds an $\epsilon$-root of $\dual{\f}$ using at most   $T(\dual{L},\epsilon,d)$ evaluations. By \Cref{obs:fixed}(b), the computed $\epsilon$-root of $\dual{\f}$ is an $\epsilon$-fixed-point of $\f$.

Therefore, $T(\dual{L},\epsilon,d) \in \Omega((\dual{L}/\epsilon)^{d-1})$, as claimed.
\end{proof}

\subsection{Efficient computation: $\delta$-continuous functions}
\label{sec:discontinuous}
We now prove a slightly stronger version of \Cref{thm:2d-sufficient}, that does not require $\f$ to be continuous.
\begin{definition}
\label{def:lc}
For any $\epsilon>0,L>0$, a function $\f: \rd\to \rd$ is called \emph{$(\epsilon,L)$-Lipschitz}
if for all $i,j\in[d]$ and $\bx\in\rd$, 
$|\bx-\bx'|\leq \epsilon/L$ implies $|\f(\bx)-\f(\bx')|\leq \epsilon$. 
\end{definition}
Note that an $(\epsilon,L)$-Lipschitz function need not be continuous;
it only needs to satisfy the Lipschitz condition for intervals of a specific length $\epsilon/L$.
We prove \Cref{thm:2d-sufficient} for any 
$(\epsilon,L)$-Lipschitz function $\f$.

We discretize the value of $f_i$  for all $i\in[d]$
as follows:
\begin{itemize}
\item If $f_i(\bx)<-\epsilon$ then we set $f_i(\bx):=-1$;
\item If $f_i(\bx)>+\epsilon$ then we set $f_i(\bx):=+1$;
\item If $|f_i(\bx)|\leq \epsilon$ then we set $f_i(\bx):=0$.
\end{itemize}
Note that the discretization procedure preserves both the monotonicity properties of $\f$ and its positive-switching properties.
Obviously, to find an $\epsilon$-root of the original $\f$, it is sufficient to find a root of the discretized $\f$.

Recall that given two points $\ba,\bb\in \rd$ with $\ba \leq \bb$,  we denote by $[\ba,\bb]$ the $d$-dimensional box in which $\ba$ and $\bb$ are two opposite corners. 
We partition $[\ba,\bb]$ into a grid of $d$-boxes with side-length $\delta:= \epsilon/L$.
The $(\epsilon,L)$-Lipschitz condition on the original $\f$ implies that the discretized value of each $f_i$ cannot have opposite signs in a single grid-cell. That is, in any $d$-box of side-length $\delta$, the values of $f_i$ on the box corners are either in $\{0,1\}$ or in $\{0,-1\}$. We call this property \emph{$\delta$-continuity}.
\begin{definition}
\label{def:discrete-continuity}
For any $\delta>0$, a function $\f:\rd\to \sign^d$ is called \emph{$\delta$-continuous} if $|\bx-\bx'|\leq \delta$ implies $|\f(\bx)-\f(\bx')|\leq 1$.
\end{definition}
From now on, 
we use $\delta$ as the approximation parameter, and do not think of $\epsilon$ or $L$ anymore. 
We look for conditions that guarantee that a root of a 
$\delta$-continuous function can be found in time $O(poly(\log(1/\delta)))$.
We first state the simple one-dimensional case:
\begin{lemma}
\label{lem:1d}
Let $f: [a,b]\to \sign$ be a positive-switching $\delta$-continuous function for some $\delta>0$. 
Then a root of $f$ can be found using $O(\log(1/\delta))$ evaluations (assuming $a,b$ are fixed).
\end{lemma}
\begin{proof}
Applying $\log_2((b-a)/\delta)$ steps of the bisection method yields an interval $[x,x+\delta]$ such that $f(x)\leq 0 \leq f(x+\delta)$. By $\delta$-continuity, at least one of $x,x+\delta$ must be a root of $f$.
\end{proof}

\subsection{Single diagonal monotonicity condition}

We make the following technical assumption for simplicity.
\begin{assumption}
\label{asm:tech}
For all $i\in[d]$, $|b_i-a_i|$ is an integer power of $2$. Moreover, $\delta$ is an integer power of $2$ (possibly with a negative exponent) and $\delta\leq |b_i-a_i|$.
\end{assumption}
\Cref{asm:tech} means that we can partition the box $[\ba,\bb]$ into a grid of $d$-cubes of side-length $\delta$, and the bisection method will only ever have to evaluate $\f$ on grid points.

\begin{theorem}
\label{thm:root-2d}
Let $\f: [\ba,\bb]\to \sign^2$ be a two-dimensional positive-switching $\delta$-continuous function
for some $\delta>0$ satisfying \Cref{asm:tech}. 
If $f_1$ is monotonically increasing with $x_1$, then a root of $\f$ can be found using $O(\log^2(1/\delta))$ evaluations.
\end{theorem}
\begin{proof}
We use \Cref{lem:1d} (the bisection method) for finding a root of a one-dimensional $\delta$-continuous switching function.
In case there are several roots, the method may return one arbitrarily but consistently, so that we can consider this as a function that maps a one-dimensional function $f$ to a single root of $f$.

By assumption, for all $x_2\in[a_2,b_2]$, $f_{1,1,x_2}$ is positive-switching in $[a_1,b_1]$.
By \Cref{lem:1d} we can find a root of 
$f_{1,1,x_2}$ using $O(\log(1/\delta))$ evaluations of $f_1$.
Define a one-dimensional function $g: [a_2,b_2]\to [a_1,b_1]$ 
such that, for all $x_2\in[a_2,b_2]$, $g(x_2)$ is the root of $f_{1,1,x_2}$ in $[a_1,b_1]$, that is returned by the bisection method. 
So for all $x_2\in[a_2,b_2]$, 
we have $f_1(g(x_2),x_2)=0$.

Define the function $h:[a_2,b_2]\to \rr$ 
such that $h(x_2) := f_2(g(x_2),x_2)$, 
that is, the value of $f_2$ at the root of $f_{1,1,x_2}$.
Because $f_2$ is positive-switching, $h(a_2)\leq 0 \leq h(b_2)$, that is, $h$ is positive-switching too.
We apply the bisection method to $h$ 
and find a point $y_2\in[a_2,b_2 - \delta]$ for which $h(y_2) \leq  0 \leq h(y_2+\delta)$.
This requires $O(\log(1/\delta))$ evaluations of $h$, each of which requires $O(\log(1/\delta))$ evaluations of $f_1$ and one evaluation of $f_2$; all in all, 
$O(\log^2(1/\delta))$ evaluations are required.

Denote $y_1 := g(y_2)$ and $z_1 := g(y_2+\delta)$. The definition of $g$ implies that $f_1(y_1,y_2)=0$ and $f_1(z_1,y_2+\delta)=0$;
the definition of $h$ implies that $f_2(y_1,y_2)=h(y_2)$ and $f_2(z_1,y_2+\delta)=h(y_2+\delta)$.
Therefore, if $h(y_2)=0$ then $(y_1,y_2)$ is a root of $\f$,
and if $h(y_2+\delta)=0$ then $(z_1,y_2+\delta)$ is a root of $\f$.
However, $h$ is not necessarily $\delta$-continuous, so it is possible that $h(y_2)=-1$ and $h(y_2+\delta)=+1$.
This case is illustrated in the following diagram:
\begin{align}
\label{pic:f2}
\begin{tikzpicture}[scale=2.5]
\node (y1y2) at (0,0) {};
\smalldot{y1y2}{north}{\shortstack{$(y_1,y_2)$\\$f_1 = 0, f_2 = -1$}};
\node (z1y22) at (4,1) {};
\smalldot{z1y22}{north}{\shortstack{$(z_1,y_2+\delta)$\\$f_1=0, f_2=+1$}};
\end{tikzpicture}
\end{align}
Note that $z_1\neq y_1$ due to the $\delta$-continuity of $f_2$. We handle here the case that $z_1>y_1$; the case that $z_1<y_1$ can be handled analogously.

We now use the monotonicity of $f_1$. For any $x_1 \in [y_1,z_1]$, the value of $f_1(x_1,y_2)$ must be in $\{0,1\}$,
and the value of $f_1(x_1,y_2+\delta)$ must be in $\{0,-1\}$. Since $f_1$ is $\delta$-continuous, there must not be adjacent values of $x_1$ for which $f_1$ attains opposite signs; so there must be a path from $(y_1,y_2)$ to $(z_1,y_2+\delta)$ on which $f_1=0$, as illustrated below:
\begin{align}
\label{pic:f1}
\begin{tikzpicture}[scale=2.5]
\node (y1y2) at (0,0) {};
\smalldot{y1y2}{north}{\shortstack{$(y_1,y_2)$\\$f_1 = 0, f_2 < 0$}};
\node (y11y2) at (1,0) {};
\smalldot{y11y2}{north}{\shortstack{$f_1 = 0$}};
\node (y12y2) at (2,0) {};
\smalldot{y12y2}{north}
{$f_1 = 0$};%
\node (y12y22) at (2,1) {};
\smalldot{y12y22}{south}{\shortstack{$f_1 = 0$}};
\node (y13y22) at (3,1) {};
\smalldot{y13y22}{south}{\shortstack{$f_1 = 0$}};
\node (z1y22) at (4,1) {};
\smalldot{z1y22}{south}{\shortstack{$(z_1,y_2+\delta)$\\$f_1=0, f_2>0$}};
\end{tikzpicture}
\end{align}

Along this path, the function $f_2$ is $\delta$-continuous and positive-switching, so it has a root and it can be found in $O(\log(1/\delta))$ evaluations. As this root of $f_2$ is found along a path on which $f_1=0$, it is also the desired root of $\f$.
\end{proof}
This completes the proof of \Cref{thm:root-2d}, which implies \Cref{thm:2d-sufficient}.

\subsection{Single ex-diagonal monotonicity  condition}
For the next theorem, we assume that the positive-switching conditions are strict, that is, 
$x_i=a_i \implies f_i(\bx)<0$ and 
$x_i=b_i \implies f_i(\bx)>0$.
This does not lose much generality, since if it does not hold, we can increase the grid by $\delta$ in each direction, and define the function values such that the strict conditions hold. 
For example, we can add the hyperplane $x_i = a_i-\delta$ and set $f_i(x_i=a_i-\delta)$ to $-1$, and 
add the hyperplane $x_i = b_i+\delta$ and set $f_i(x_i=b_i+\delta)$ to $+1$.
As the original $\f$ is positive-switching, 
the addition of these hyperplanes does not harm the $\delta$-continuity property, as well as any monotonicity property.

\begin{theorem}
\label{thm:root-2d-exdiagonal}
Let $\f: [\ba,\bb]\to \sign^2$ be a two-dimensional strictly-positive-switching $\delta$-continuous function
for some $\delta>0$ satisfying \Cref{asm:tech}. 
If $f_1$ is monotonically-decreasing with $x_2$, then a root of $\f$ can be found using $O(\log^2(1/\delta))$ evaluations.
\end{theorem}
\begin{proof}
For any $x_1\in[a_1,b_1]$, we consider the function $f_{1,2,x_1}$,
which is $f_1$ considered as a function of the single variable $x_2$, when $x_1$ is fixed.
The monotonicity condition implies that there are only three possibilities regarding the value of $f_{1,2,x_1}$ on the endpoints $a_2,b_2$:
\begin{itemize}
\item $f_{1,2,x_1}(a_2) \geq 0$ and $f_{1,2,x_1}(b_2) \leq 0$;
\item $f_{1,2,x_1}(a_2) < 0$ (which implies $f_{1,2,x_1}(b_2) < 0$);
\item $f_{1,2,x_1}(b_2) > 0$ (which implies $f_{1,2,x_1}(a_2) >0$).
\end{itemize}
In the former case, $f_{1,2,x_1}$ is negative-switching in $[a_2,b_2]$. Since it is also $\delta$-continuous, we can use the bisection method to find a root of it.
Define a one-dimensional function $g: [a_1,b_1]\to [a_2,b_2]$  as follows:
\begin{align*}
g(x_1) := 
\begin{cases}
\text{The root of $f_{1,2,x_1}$ returned by the bisection method} & \text{ if $f_{1,2,x_1}(b_2) \leq 0 \leq f_{1,2,x_1}(a_2)$}
\\
a_2 & \text{ if $f_{1,2,x_1}(a_2)<0$}
\\
b_2 & \text{ if $f_{1,2,x_1}(b_2)>0$}
\end{cases}
\end{align*}	
Because $f_1$ is strictly-positive-switching, $f_{1,2,a_1}(\cdot)< 0$, so $g(a_1) = a_2$.
Similarly, 
$f_{1,2,b_1}(\cdot)> 0$, so $g(b_1) = b_2$.

Now define a function $h:[a_1,b_1]\to \rr$ 
by $h(x_1) := f_2(x_1, g(x_1))$. 
Because $f_2$ is strictly-positive-switching, 
$h(a_1) = f_2(a_1,g(a_1)) = f_2(a_1,a_2)< 0$,
and $h(b_1) = f_2(b_1,g(b_1))=f_2(b_1,b_2)>0$. 
Therefore, $h$ is positive-switching too.
We apply the bisection method to $h$ 
and find a point $y_1\in[a_1,b_1 - \delta]$ for which $h(y_1) \leq  0 \leq h(y_1+\delta)$.
This requires 
$O(\log^2(1/\delta))$ evaluations.

Denote $y_2 := g(y_1)$ and $z_2 := g(y_1+\delta)$. 
The definition of $h$ implies that $f_2(y_1,y_2)=h(y_1)\leq 0$ and $f_2(y_1+\delta,z_2)=h(y_1+\delta)\geq 0$.
We now consider several cases regarding the values of $f_1$ at these points.

\underline{Case 1:}
Both $(y_1,y_2)$ and $(y_1+\delta,z_2)$ are roots of $f_1$, so
$f_1(y_1,y_2)=0$ and $f_1(y_1+\delta,z_2)=0$.
Then we are in the same situation as in the proof of \Cref{thm:root-2d}, except that figures \eqref{pic:f2} and \eqref{pic:f1} are rotated by $90^\circ$. By a similar argument, a root of $\f$ can be found on a path between $(y_1,y_2)$ and $(y_1+\delta, z_2)$.

\underline{Case 2:} 
$(y_1+\delta,z_2)$ is a root of $f_1$,
but $(y_1,y_2)$ is not a root of $f_1$,
so $y_2 = g(y_1)$ is not a root of $f_{1,2,y_1}$.
By definition of $g$, this is possible only in two cases:
\begin{itemize}
\item $f_1(y_1,a_2)<0$ and $y_2=a_2$ (which implies $f_2(y_1,y_2)<0$ by strict-positive-switching of $f_2$);
\item $f_1(y_1,b_2)>0$ and $y_2=b_2$ (which implies 
$f_2(y_1,y_2)>0$ by strict-positive-switching of $f_2$).
\end{itemize}
But the latter option contradicts the inequality $f_2(y_1,y_2)\leq 0$, so the former option must hold.
Moreover, the monotonicity of $f_1$ 
implies $f_1(y_1,x_2)<0$ for all $x_2\in[a_2,b_2]$
and $f_1(y_1+\delta,x_2)\geq 0$ for all $x_2\in[a_2,z_2]$;
combined with the $\delta$-continuity of $f_1$, this implies 
$f_1(y_1+\delta,x_2)= 0$ for all $x_2\in[a_2,z_2]$.
The situation is illustrated in the following diagram:

\begin{align*}
\begin{tikzpicture}[scale=2.5]
\node (y1y2) at (0,0) {};
\smalldot{y1y2}{east}{\shortstack{$(y_1,y_2=a_2)$\\$f_1 < 0, f_2 < 0$}};
\node (y11y2) at (0,1) {};
\smalldot{y11y2}{east}{\shortstack{$f_1 < 0$}};
\node (y12y2) at (0,2) {};
\smalldot{y12y2}{east} {$f_1 < 0$};
\node (y12y2) at (0,3) {};
\smalldot{y12y2}{east} {$f_1 < 0$};
\node (y12y22) at (1,0) {};
\smalldot{y12y22}{west}{\shortstack{$f_1 = 0, f_2<0$}};
\node (y12y22) at (1,1) {};
\smalldot{y12y22}{west}{\shortstack{$f_1 = 0$}};
\node (y12y22) at (1,2) {};
\smalldot{y12y22}{west}{\shortstack{$f_1 = 0$}};
\node (z1y22) at (1,3) {};
\smalldot{z1y22}{west}{\shortstack{$(y_1+\delta,z_2)$\\$f_1=0, f_2 \geq 0$}};
\end{tikzpicture}
\end{align*}
On the segment connecting $(y_1+\delta,a_2)$ and $(y_1+\delta,z_2)$, the function $f_2$ is positive-switching, so it has a root that can be computed using $O(\log(1/\delta))$ evaluations. As $f_1=0$ along this segment, this is the desired root of $\f$.

\underline{Case 3:}
$(y_1,y_2)$ is a root of $f_1$,
but $(y_1+\delta,z_2)$ is not a root of $f_1$.
By arguments similar to Case 2, this implies
$f_1(y_1+\delta,b_2)>0$ and $z_2=b_2$, which implies 
$f_2(y_1+\delta,z_2)>0$ since $f_2$ is positive-switching.
The situation is illustrated in the following diagram:

\begin{align*}
\begin{tikzpicture}[scale=2.5]
\node (y1y2) at (0,0) {};
\smalldot{y1y2}{east}{\shortstack{$(y_1,y_2)$\\$f_1 = 0, f_2 \leq 0$}};
\node (y11y2) at (0,1) {};
\smalldot{y11y2}{east}{\shortstack{$f_1 = 0$}};
\node (y12y2) at (0,2) {};
\smalldot{y12y2}{east} {$f_1 = 0$};
\node (y12y2) at (0,3) {};
\smalldot{y12y2}{east} {$f_1 = 0, f_2>0$};
\node (y12y22) at (1,0) {};
\smalldot{y12y22}{west}{\shortstack{$f_1 > 0$}};   
\node (y12y22) at (1,1) {};
\smalldot{y12y22}{west}{\shortstack{$f_1 > 0$}};
\node (y12y22) at (1,2) {};
\smalldot{y12y22}{west}{\shortstack{$f_1 > 0$}};
\node (z1y22) at (1,3) {};
\smalldot{z1y22}{west}{\shortstack{$(y_1+\delta,b_2)$\\$f_1>0, f_2 > 0$}};
\end{tikzpicture}
\end{align*}
On the segment connecting $(y_1,y_2)$ and $(y_1,b_2)$, the function $f_2$ is positive-switching, so it has a root that can be computed using $O(\log(1/\delta))$ evaluations. As $f_1=0$ along this segment, this is the desired root of $\f$.

\underline{Case 4:}
Both $(y_1,y_2)$ and $(y_1+\delta,z_2)$ are not roots of $f_1$.
This case is impossible, since by the arguments of Cases 2 and 3, this would imply $f_1(y_1,x_2)<0$ and $f_1(y_1+\delta,x_2)>0$ for all $x_2\in[a_2,b_2]$, but this would contradict $\delta$-continuity.
\end{proof}

This completes the proof of \Cref{thm:root-2d-exdiagonal}, which implies \Cref{thm:2d-sufficient-exdiagonal}.

\subsection{Alternative switching condition}
Consider the following variant of \Cref{def:switching-cube}:
\begin{definition}
\label{def:sum-switching}
A $d$-dimensional function $\f: [\ba,\bb]\to \rd$ is called \emph{sum-switching} if for every $i\in[d]$ and every $\bx\in[\ba,\bb]$, the following hold:
\begin{align*}
x_i=a_i &\implies f_i(\bx)\leq 0;
\\
x_i=b_i &\implies \sum_{j=1}^i f_j(\bx)\geq 0.
\end{align*}
\end{definition}
We show that \Cref{thm:root-2d} holds whenever $\f$ is 
sum-switching.
\begin{lemma}\label{lem:sum-switching}
Let $\f: [\ba,\bb]\to \sign^2$ be a 
$\delta$-continuous discretization of a two-dimensional sum-switching function $\f':[\ba,\bb]\to \rb$. Then, $\f$ is also sum-switching.
\end{lemma}


\begin{proof}
We only need to show that the condition holds for $i=2$ and $x_2 = b_2$, as the other cases hold trivially. Let $\epsilon$ denote the approximation factor that was used to construct $\f$ from $\f'$. Since $\f'$ is sum-switching, we have $f_1'(x_1,b_2) + f_2'(x_1,b_2)\geq 0$ for all $x_1$. Assume towards a contradiction that $f_1(x_1,b_2) + f_2(x_1,b_2) < 0$ for some $x_1$. Then, it must be that $f_1'(x_1,b_2) < - \epsilon$ or $f_2'(x_1,b_2) < - \epsilon$. But if $f_1'(x_1,b_2) < - \epsilon$, then by the sum-switching condition we must have $f_2'(x_1,b_2) \geq -f_1'(x_1,b_2) > \epsilon$. As a result, $f_2(x_1,b_2) = 1$, and since $f_1(x_1,b_2) \geq -1$, we have $f_1(x_1,b_2) + f_2(x_1,b_2) \geq 0$, a contradiction. The case where $f_2'(x_1,b_2) < - \epsilon$ is handled analogously.
\end{proof}

\begin{theorem}
\label{thm:root-2d-sum}
Let $\f: [\ba,\bb]\to \sign^2$ be a 
$\delta$-continuous discretization of a two-dimensional sum-switching function,
for some $\delta>0$ satisfying \Cref{asm:tech}.
If $f_{1}$ is monotonically-increasing with $x_1$, then 
a root of $\f$ can be found using $O(\log^2(1/\delta))$ evaluations.
\end{theorem}
\begin{proof}
The proof is almost identical to that of \Cref{thm:root-2d}. By \Cref{lem:sum-switching}, we know that $\f$ is also sum-switching.
For $f_1$, the sum-switching condition is identical to the positive-switching condition, so we can define  functions $g$ and $h$ as in the original proof.
For $f_2$, sum-switching implies that $f_2(x_1,a_2)\leq 0$ for all $x_1$, which implies $h(a_2)\leq 0$.
Sum-switching also implies that, for all $x_1$, 
$f_1(x_1,b_2) + f_2(x_1,b_2)\geq 0$.
Since $f_1(g(b_2),b_2) = 0$ by definition of $g$,  $f_2(g(b_2),b_2)\geq 0$ must hold, so $h(b_2)\geq 0$.
Hence, $h$ is positive-switching, and the rest of the proof remains valid.
\end{proof}

\section{Finding approximate roots: $d\geq 3$ dimensions}
\label{sec:roots-dd}

\subsection{$d^2-2$ monotonicity conditions are insufficient}
In contrast to the positive result for $d=2$, we prove a negative result for $d\geq 3$.

\begin{proposition*}[= \Cref{prop:dd-insufficient}]
For every $\epsilon>0, L>0, d\geq 3$,
let $T(L,\epsilon,d)$ be the worst-case number of evaluations required by an algorithm that finds an $\epsilon$-root for every $L$-Lipschitz-continuous function that is positive-switching on some $d$-cube, and satisfies $d^2-2$ monotonicity conditions
(including the $d$ diagonal conditions and  $d^2-d-2$ ex-diagonal conditions). Then $T(L,\epsilon,d) \in \Omega(L/\epsilon)$.
\end{proposition*}

\begin{proof}

The proof is by reduction from a general 2-dimensional function (\Cref{prop:2d-insufficient}) to a monotone 3-dimensional function.

Let $\g$ be any $1$-Lipschitz-continuous function
that is positive-switching on $[-1,1]^2$.
Construct a function $\f: [-1,1]^3\to \mathbb{R}^3$ as follows.
\begin{align*}
f_1(\bx) &= g_1(x_1,x_3) + 2\cdot (x_1 - \operatorname{trunc}(2x_2-x_3))
\\
f_2(\bx) &= 2x_2 - x_1 - x_3
\\
f_3(\bx) &= g_2(x_1,x_3) + 2\cdot (x_3 - \operatorname{trunc}(2x_2-x_1))
\end{align*}
where trunc denotes truncation to $[-1,1]$, that is,  $\operatorname{trunc}(x) = \min(1,\max(-1,x))$.
We prove several claims on $\f$.

\begin{claim}
$\f$ is $O(1)$- Lipschitz-continuous on $[-1,1]^3$.
\end{claim}
\begin{proof}
This is obvious, as $\f$ is a sum of $O(1)$-Lipschitz-continuous functions.
\end{proof}

\begin{claim}
$\f$ is positive-switching on $[-1,1]^3$.
\end{claim}
\begin{proof}
We have to prove that each $f_i$ is positive-switching with respect to $x_i$.

For $f_2$, this follows from the fact that $x_1$ and $x_3$ are between $-1$ and $1$.

For $f_1$ and $f_3$, this follows from the fact that $g_1$ and $g_2$ are positive-switching, and the outcome of $\operatorname{trunc}$ is between $-1$ and $1$.
\end{proof}

\begin{claim}
$\f$ satisfies $7$ monotonicity conditions, including all $3$ diagonal conditions and $4$ out of $6$ ex-diagonal conditions.
\end{claim}
\begin{proof}
To see that $f_1$ is increasing with $x_1$, note that $f_1(\bx) = 2 x_1 + g(x_1,x_3) + \text{const}$.
Whenever $x_1$ increases by $\epsilon$, $2 x_1$ increases by $2\cdot \epsilon$ and $g$ decreases by at most $\epsilon$ since it is $1$-Lipschitz-continuous. Therefore, $f_1$ increases.
Moreover, $f_1$ is decreasing with $x_2$, as $-\operatorname{trunc}$ is decreasing.

By similar arguments, $f_3$ is increasing with $x_3$ and decreasing with $x_2$.

Obviously, $f_2$ is increasing with $x_2$ and decreasing with $x_1,x_3$. Overall, all $3$ diagonal conditions and $4$ ex-diagonal conditions are satisfied.
\end{proof}

\begin{claim}
If $\bx=(x_1,x_2,x_3)$ is an $\epsilon$-root of $\f$, then $(x_1,x_3)$ is a $3 \epsilon$-root of $\g$.
\end{claim}
\begin{proof}
$\bx$ is an $\epsilon$-root of $\f$ implies
$|f_2(\bx)|\leq \epsilon$, 
so 
\begin{align}
\label{eq:x1}
x_1-\epsilon &\leq  (2 x_2 - x_3)\leq x_1 +\epsilon
\\
\label{eq:x3}
 x_3-\epsilon &\leq  (2 x_2 - x_1)\leq x_3+\epsilon 
\end{align}
Since $x_1$ and $x_3$ are in $[-1,1]$, the same inequalities hold when truncating the central expressions $2x_2-x_3$ and $2x_2-x_1$.

$\bx$ is an $\epsilon$-root of $\f$ implies
also $|f_1(\bx)|\leq \epsilon$, 
so
\begin{align*}
2\cdot(x_1 - \operatorname{trunc}(2 x_2 - x_3)) - \epsilon
\leq
|g_1(x_1,x_3)| \leq 2\cdot(x_1 - \operatorname{trunc}(2 x_2 - x_3)) + \epsilon.
\end{align*}
Substituting \eqref{eq:x1} gives
\begin{align*}
-2\cdot\epsilon - \epsilon
\leq 
|g_1(x_1,x_3)|
\leq
2\cdot \epsilon + \epsilon.
\end{align*}
Similarly, using \eqref{eq:x3} and $|f_3(\bx)|\leq \epsilon$ we get
\begin{align*}
-2\cdot\epsilon - \epsilon
\leq 
|g_1(x_1,x_3)|
\leq
2\cdot \epsilon + \epsilon.
\end{align*}
Therefore, $(x_1,x_3)$ is a $3 \epsilon$-root of $\g$.
\end{proof}

By \Cref{prop:2d-insufficient}, any algorithm for finding a
$3\epsilon$-root for general 2-dimensional functions may require $\Omega(1/\epsilon)$ evaluations. 
As a result, the above reduction implies that finding an $\epsilon$-root of the $O(1)$-Lipschitz function $\f$ requires $\Omega(1/\epsilon)$ evaluations. 
Scaling $\f$ by any $L>0$ yields the desired bound $\Omega(L/\epsilon)$, as an $\epsilon$-root of the scaled function is an $\epsilon/L$-root of the original function.
This completes the proof of the Proposition for $d=3$.

To prove the Proposition for any $d\geq 3$, simply add functions $f_i$ for $i\geq 4$, that are identically $0$. The functions are trivially Lipschitz-continuous,  positive-switching, and monotonically increasing with $x_j$ for any $j\in[d]$.
Moreover, for any $i\in[d]$ and $j\geq 4$, $f_i$ is trivially monotonically increasing with $x_j$. Therefore, $\f$ satisfies $d^2-2$ monotonicity conditions, and the same proof applies.
\end{proof}

\begin{remark}
     The lower bound for $d\geq 4$ could possibly be strengthened to $\Omega((L/\epsilon)^{d-2})$, but we do not pursue it as our main goal is to distinguish between the polynomial and  exponential cases. 
\end{remark}

\subsection{$d^2-d$ ex-diagonal conditions are sufficient}

\begin{theorem*}[= \Cref{thm:dd-sufficient}]
Let function $\f: [\ba,\bb]\to \sign^d$ be 
positive-switching, such that, for all $\bx\in[\ba,\bb]$,  $f_i$ is decreasing with $x_j$ all $i\neq j\in[d]$. 
Then, 
for any $\delta>0$
satisfying \Cref{asm:tech},
if $\f$ is $\delta$-continuous, then a root of $\f$ can be found using 
$O(\log^{\lceil (d+1)/2\rceil}(1/\delta))$ evaluations.
\end{theorem*}

For the proof, we reduce the problem of finding a root to the problem of finding a \emph{Tarski fixed-point}.
The term refers to a fixed point whose existence is guaranteed by the following theorem.
\begin{theorem*}[Knaster--Tarski Theorem]
Let $(Z, \succeq)$ be a complete lattice (a partially-ordered set in which all subsets have a supremum and an infimum).
Let $\h:Z\to Z$ be a function on $Z$ that is  order-preserving with respect to $\succeq$.
Then $\h$ has at least one fixed point in $Z$.
\end{theorem*}

We apply the Knaster--Tarski theorem to the case when $(Z, \succeq)$ is a \emph{component-wise lattice}. This means that $Z$ is a subset of $\mathbb{R}^d$ (particularly, the $\delta$-grid defined earlier), and $\succeq$ is the partial order defined by: $\bx' \succeq \bx$ iff $x'_i\geq x_i$ for all $i\in[d]$.

\begin{proof}[Proof of \Cref{thm:dd-sufficient}]
Define a new function $\h: Z\to Z$ by
\begin{align*}
    \h(\bx) := \bx - \f(\bx)\cdot \delta.
\end{align*}
Note that $\bx$ is a grid point, so each of its coordinates is a multiple of $\delta$.
By definition of $\f$, each coordinate of $\f(\bx)$ is 
in $\{-1,0,1\}$.
Therefore, each coordinate of $\h(\bx)$ is a multiple of $\delta$.

\begin{claim}
\label{claim:g-order-preserving}
$\h$ is order-preserving in the component-wise lattice. 
\end{claim}
\begin{proof}
It is sufficient to prove the claim for adjacent grid-points.
Let $\bx\in [\ba,\bb]$ by a grid-point 
and let $\bx'$ be an adjacent grid-point define by  $x'_j = x_j + \delta$ for some $j\in[d]$, and $x'_i = x_i$ for all $i\neq j$.
We have to prove that $h_i(\bx')\geq h_i(\bx)$ for all $i\in[d]$.

For all $i\neq j$, the ex-diagonal monotonicity conditions on $\f$ imply that $f_i(\bx')\leq f_i(\bx)$, so $h_i(\bx') = x_i'-f_i(\bx')\cdot \delta = x_i - f_i(\bx')\cdot \delta \geq x_i - f_i(\bx)\cdot \delta = h_i(\bx)$.

For the $j$-th coordinate, we have
$h_j(\bx') = x_j' - f_j(\bx')\cdot \delta = x_j + \delta- f_j(\bx')\cdot \delta$.
Because $f_j$ changes by at most $1$ between adjacent grid-points, $h_j(\bx') \geq x_j - f_j(\bx)\cdot \delta = h_j(\bx)$.
Therefore, $\h$ is order-preserving as claimed.
\end{proof}

\begin{claim}
\label{claim:g-maps-box-to-itself}
$\h$ maps the $\delta$-grid on $[\ba,\bb]$ to itself. 
\end{claim}
\begin{proof}
For all $i\in[d]$,
if $x_i= a_i$, then
\begin{align*}
h_i(\bx) &= x_i - f_i(\bx)\cdot \delta
\\
& = a_i - f_i(\bx)\cdot \delta
\\
& \geq a_i, && \text{since $f_i(\bx)\leq 0$.}
\end{align*}
Similarly, if $x_i= b_i$, then $h_i(\bx)\leq b_i$ since $f_i(\bx)\geq 0$.

\Cref{claim:g-order-preserving} implies that $f_{i,i,\bx}$ is increasing, so $a_i\leq f_i(\bx)\leq b_i$ for all $\bx$ in the grid. 
Since both the grid-points and the values of $\f$ are integer multiples of $\delta$, $\h$ maps the grid to itself.
\end{proof}

\citet{chen2022improved} proved that, for a component-wise lattice $(Z,\succeq)$, a Tarski fixed-point of $\h: Z\to Z$ can be found using $O(\log^{\lceil (d+1)/2\rceil} (N))$ evaluations of $\h$, where $N$ is the maximum number of points in each dimension; in our 
lattice, $N\in O(1/\delta)$. 

Let $\bx$ be a fixed point of $\h$, so $\h(\bx)=\bx$.
Then $\f(\bx) = (\bx - \h(\bx))/\delta = 0$, so $\bx$ is the desired root of $\f$.
\end{proof}

\begin{proof}[Alternative proof]
To better understand what happens ``under the hood'' of the reduction, 
we present a direct algorithm for finding a root of $\f$, with a slightly worse run-time of $O(\log^d(1/\delta))$ evaluations (but still polynomial in the accuracy).

The proof is by induction on $d$, and follows the lines of the algorithm of \cite{dang2020computations}.

For the base $d=1$ no monotonicity conditions are assumed, but we can still use \Cref{lem:1d}, as it does not require monotonicity.%
\footnote{
We could also start with $d=2$ and apply \Cref{thm:root-2d-exdiagonal}.
}
By \Cref{lem:1d}, a root of $\f$ can be found using $O(\log(1/\delta))$ evaluations. 

For the induction step, we assume that the statement holds for $d-1$.
Let $y_d := (a_d+b_d)/2 = $ the middle value of $x_d$.
By \Cref{asm:tech}, $y_d$ is a coordinate of one of the grid hyperplanes.
~
Consider the $(d-1)$-dimensional function $\g^{-}$, defined by 
\begin{align*}
g^{-}_j(x_1,\ldots,x_{d-1}) := f_j(x_1,\ldots,x_{d-1},y_d),
\text{ for all $j\in[d-1]$.}
\end{align*}
Let $\ba^- := (a_1,\ldots,a_{d-1})$
and $\bb^- := (b_1,\ldots,b_{d-1})$.
The function $\g^-$ is positive-switching
on $[\ba^-,\bb^-]$ and satisfies all ex-diagonal monotonicity conditions,
so by the induction assumption it has a root 
that can be found using $O(\log^{d-1}(1/\delta))$ evaluations.
Denote this root by $(y_1,\ldots,y_{d-1})$,
and define $\by := (y_1,\ldots,y_{d-1},y_d)$.
So $f_j(\by) = 0$ for all $j\in[d-1]$.

We consider three cases, depending on the sign of $f_d(\by)$.

\paragraph{Case 1:}
$f_d(\by)= 0$. Then $\by$ is a root of $\f$, and we are done.

\paragraph{Case 2:}
$f_d(\by) < 0$.
Consider the box 
$[\by,\bb]$ containing all $\bx \in[\ba,\bb]$ such that $x_i\geq y_i$ for all $i\in [d]$.
Due to the ex-diagonal monotonicity conditions, 
we have, for all 
$\bx\in [\by,\bb]$,
\begin{align*}
f_j(x_1,\ldots,x_{j-1},y_j,x_{j+1},\ldots,x_d) &\leq f_j(\by) = 0 && \forall j\in[d-1];
\\
f_d(x_1,\ldots,x_{d-1},y_d) &\leq f_d(\by) < 0
\end{align*}
Therefore, $\f$ is positive-switching on $[\by,\bb]$, and we can recursively apply the same algorithm.

\paragraph{Case 3:}
$f_d(\by) > 0$.
Consider the box 
$[\ba,\by]$ containing all $\bx \in[\ba,\bb]$ such that $x_i\leq y_i$ for all $i\in [d]$.
Due to the ex-diagonal monotonicity conditions, 
we have, for all 
$\bx\in [\by,\bb]$,
\begin{align*}
f_j(x_1,\ldots,x_{j-1},y_j,x_{j+1},\ldots,x_d) &\geq f_j(\by) = 0 && \forall j\in[d-1];
\\
f_d(x_1,\ldots,x_{d-1},y_d) &\geq f_d(\by) > 0
\end{align*}
Therefore, $\f$ is positive-switching on $[\ba,\by]$, and we can recursively apply the same algorithm. 

In both cases 2 and 3, the size of the box in dimension $d$ is halved at each recursion round. Therefore, after at most $\log(1/\delta)$ iterations, the box size in dimension $d$ is at most $\delta$, and the recursion must end at Case 1 with a root of $\f$.
Each iteration requires finding a root of a $(d-1)$-dimensional function, which by the induction assumption takes $O(\log^{d-1}(1/\delta))$ evaluations.
All in all, $O(\log^{d}(1/\delta))$ evaluations suffice for finding the desired root of $\f$.
\end{proof}

\subsection{All $d$ switching conditions are necessary}

\begin{proposition*}[=\Cref{prop:switcing-necessary}]
For every $\epsilon>0, L>0, d\geq 2$,
let $T(L,\epsilon,d)$ be the worst-case number of evaluations required by an algorithm that finds an $\epsilon$-root for every $L$-Lipschitz-continuous function that has an exact root and satisfies $d-1$ positive-switching conditions and all $d^2$ monotonicity conditions. Then $T(L,\epsilon,d) \in \Omega(L/\epsilon)$.
\end{proposition*}
\begin{proof}
The proof is similar to that of \Cref{prop:dd-insufficient}, using a reduction from a general (neither switching nor monotone) 1-dimensional function.

We first prove the proposition for $d=2$.
Let $g:[-1,1]\to \mathbb{R}$ be any $1$-Lipschitz-continuous function that has a root.
Construct a function $\f: [-1,1]^2\to \mathbb{R}^2$ as follows.
\begin{align*}
f_1(\bx) &= g(x_1) + 2\cdot (x_1 - x_2)
\\
f_2(\bx) &= x_2 - x_1
\end{align*}
Clearly, $\f$ is $O(1)$- Lipschitz-continuous on $[-1,1]^2$, and satisfies one switching condition, namely, $f_2$ is positive-switching with $x_2$.

\begin{claim}
$\f$ satisfies all $4$ monotonicity conditions.
\end{claim}
\begin{proof}
It is obvious that $f_1$ is decreasing with $x_2$, and  $f_2$ is increasing with $x_2$ and decreasing with $x_1$.
To see that $f_1$ is increasing with $x_1$, 
suppose $x_1$ increases by $\epsilon$. Then $2 x_1$ increases by $2\cdot \epsilon$ and $g$ decreases by at most $\epsilon$ since it is $1$-Lipschitz-continuous, so overall $f_1$ increases.
\end{proof}
\begin{claim}
$\f$ has an exact root.
\end{claim}
\begin{proof}
By assumption, $g$ has a root, say $x$. Then $(x,x)$ is a root of $\f$.
\end{proof}
\begin{claim}
If $\bx=(x_1,x_2)$ is an $\epsilon$-root of $\f$, then $x_1$ is a $3 \epsilon$-root of $g$.
\end{claim}
\begin{proof}
$\bx$ is an $\epsilon$-root of $\f$ implies
$|x_2-x_1| = |f_2(\bx)|\leq \epsilon$,
and also $|g_1(x) + 2(x_1-x_2)| = |f_1(\bx)|\leq \epsilon$.
Combining the two inequalities implies 
$|g_1(x)| \leq \epsilon + 2|x_2-x_1|
\leq 3\epsilon$, so $x_1$ is a $3 \epsilon$-root of $g$.
\end{proof}
By the above claims, an $\epsilon$-root of $\f$ can be found using $T(L,\epsilon,d)$ evaluations, and it yields a $3\epsilon$-root of $g$.
By \citet{sikorski1984optimal}, any algorithm for deciding existence of a $3\epsilon$-root for general 1-dimensional 1-Lipschitz functions that have a root may require $\Omega(1/\epsilon)$ evaluations. 
As a result, the above reduction implies that finding an $\epsilon$-root of the $O(1)$-Lipschitz function $\f$ requires $\Omega(1/\epsilon)$ evaluations. 
Scaling $\f$ by any $L>0$ yields the desired bound $\Omega(L/\epsilon)$, as an $\epsilon$-root of the scaled function is an $\epsilon/L$-root of the original function.
This completes the proof for $d=2$.

To prove the Proposition for any $d\geq 2$,  add functions $f_i$ for $i\geq 3$, that are identically $0$. The functions are Lipschitz-continuous and satisfy all switching and monotonicity conditions.
\end{proof}

\section{Envy-free cake-cutting}
\label{sec:envy-free}
A \emph{group cake-cutting instance} consists of a cake represented by the interval $[0,1]$, and some $n$ agents with different preferences over pieces of cake. 
The goal is to partition the cake into $m$ connected pieces (intervals), and partition the agents into $m$ groups with prespecified cardinalities $k_1,\ldots,k_m$ (with $\sum_{j=1}^m k_j = n$), and assign each piece to a group.

Agents may have different preferences over pieces of cake, which are represented by value functions. For each agent $i$ and $a,b\in [0,1]$ with $a\leq b$, we denote by $v_i(a,b)$ the value assigned by $i$ to the interval $[a,b]$.
We make the following assumptions on the value functions.
\begin{itemize}
\item All $v_i$ are continuous functions.
\item All agents value empty pieces at zero: $v_i(a,a)=0$ for all $a\in[0,1]$. 
\item In any cake-partition, each agent values at least one piece positively (this is sometimes called the ``hungry players'' assumption).
\item The agents' preferences are \emph{monotone} --- each agent weakly prefers a piece to its subsets:
$v_i(a,b)\leq v_i(a',b')$ whenever $a'\leq a\leq b\leq b'$.
\item $v_i$ are accessible by \emph{evaluation queries}: for each $a,b$, it is possible to get $v_i(a,b)$ in unit time.%
\footnote{
Other cake-cutting algorithms require \emph{cut} or \emph{mark} queries; we do not need such queries here.
}
\end{itemize}

An allocation is called \emph{envy-free} if each agent (weakly) prefers his group's piece to the other pieces. Formally, denoting the piece of group $j$ by $[a_j,b_j]$, envy-freeness requires that $v_i(a_j,b_j)\geq v_i(a_{j'},b_{j'})$ for all $i$ in group $j$ and for any other group $j'$.

Given an approximation factor $r\in(0,1)$, an \emph{$r$-near envy-free allocation} is an allocation in which, for each agent $i\in[n]$, if the partition lines are moved by at most $r$, then the agent does not envy.
We now present a general reduction from near-envy-free cake-cutting to approximate root-finding.
The construction generalizes the one by \citet{igarashi2022envy}.


\begin{theorem*}[= \Cref{thm:ef-to-root}]
For every integer $d\geq 2$ and real $r>0$, 
for every instance of group cake-cutting with $n$ agents and $m=d+1$ groups where all agents have monotone preferences,
there exists a function $\f:[0,1]^d \to\rd$ with the following properties:
\begin{enumerate}
\item  
\label{prop:evaluate} 
One evaluation of $\f$ requires $O(d^2 \cdot n)$ queries of the agents' preferences.
\item  
\label{prop:root} 
Every $1/(2\cdot d^2)$-root of $\f$ corresponds to an $r$-near envy-free allocation.
\item  
\label{prop:lipschitz} 
$\f$ is Lipschitz-continuous with constant $n/r$.
\item 
\label{prop:switching} 
$\f$ is sum-switching in $[0,1]^d$ (see \Cref{def:sum-switching}).
\item 
\label{prop:monotone} 
$\f$  is alternating-monotone in $[0,1]^d$  (see \Cref{def:alternating-monotone}).
\end{enumerate}
\end{theorem*}
\begin{proof}
We define $\f$ on the unit $d$-cube $[0,1]^d$, where $d=m-1$.
We form a grid on $[0,1]^d$ where the cells  are $d$-cubes of side-length $r$.
With each grid-point $(x_1,\ldots,x_d)$ we associate  a cake-partition in which, for each $i\in[d]$, cut $i$ is made at $\max_{j\leq i} x_j$. That is, the first cut is made at $x_1$, the second cut is made at $\max(x_1,x_2)$, and so on. These $d$ cuts induce a partition of the cake into $d+1=m$ contiguous pieces.

For each such cake-partition, we ask each agent to evaluate all pieces, and choose a piece with a highest value. In case of a tie, the agent must choose a piece that is not empty; the monotonicity of the agents' preferences implies that every agent always has a highest-valued non-empty piece. If the agent has two or more highest-valued non-empty pieces, he must choose one of them in a consistent way (e.g., by index).

To define the function $\f$, we first define an auxiliary  function $\g:[0,1]^d\to \rda$ as follows. 
On each 
grid-point 
$\bx$, 
for every $i\in[d+1]$, $g_i(\bx)$ equals the number of agents whose preferred piece at the associated partition is piece $i$. 
Note that $\sum_{i=1}^{d+1} g_i(\bx)=n$.

\def\simplex{}

To define $\g$ on the entire box $[0,1]^d$, 
\ifdefined\simplex
we apply to each grid cell the standard triangulation of the cube,%
\footnote{
In the standard triangulation of the cube,
a simplex is constructed by starting from the corner of the cube with the smallest coordinates. Then we fix a permutation, and we increase each coordinate by $r$ in the order prescribed by the permutation. Thus, every permutation yields a different simplex. This subdivides the cube into $d!$ simplices. In particular, every simplex contains the corner of the cube with the smallest coordinates, and the corner opposite to that (i.e., with the largest coordinates).
Note that this triangulation does not add new vertices.
}
and extend $\g$ affinely in each triangulation simplex. 
\fi
\ifdefined\cell
for each point $\bx$ and $j\in[d]$, define $x_j^-$ the largest grid value smaller than $x_j$, and by $x_j^+$ the smallest grid value larger than $x_j$. Define $t_j := (x_j - x_j^-)/r = $  the relative distance between $x_j$ to $x_j^-$. 
Define $\g(\bx)$ as the weighted average of the $2^d$ grid points in the cell containing $\bx$, with weights determined by the $t_j$, that is:
\begin{align*}
\g(\bx) := t_1\cdots t_d \cdot \g(x_1^+,\ldots,x_d^+) + \ldots + (1-t_1)\cdots (1-t_d) \cdot \g(x_1^-,\ldots,x_d^-).
\end{align*}
\fi
Note that, on cell boundaries, the value of $\g$ is the same regardless of what cell is used for the computation. Note also that $\sum_{i=1}^{d+1} g_i(\bx)=n$ still holds.

The function $\f$ is defined by:
\begin{align*}
f_i(\bx) := g_i(\bx)-k_i  && \forall i\in[d],   \forall \bx\in[0,1]^d,
\end{align*}
where $k_i$ is the required number of agents in group $i$.
The value of $g_{d+1}$ is discarded, but no information is lost as $g_{d+1}(\bx) = n-\sum_{i=1}^{d} g_i(\bx)$.  Note that $\sum_{i=1}^{d+1} f_i(\bx)=0$.

We now prove the properties of $\f$.

\paragraph{Property \ref{prop:evaluate}.}
Evaluating $\f$ in any  grid-point  requires to ask each of the $n$ agents to evaluate each of the $d+1$ pieces.
Evaluating $\f$ in any internal point requires to evaluate $\f$ on the 
\ifdefined\simplex
$d+1$ corners of the containing triangulation simplex. Overall, 
at most $(d+1)^2\cdot n$ evaluations are required.
\fi
\ifdefined\cell
$2^d$ corners of the containing grid cell. Overall, 
at most $(d+1)\cdot 2^d \cdot n$ evaluations are required.
\alex{If we end up using the cell approach, we should not forget to change the theorem statement to state the correct number of queries.}
\fi

\paragraph{Property \ref{prop:root}.}
Let $\bx$ be a $1/(2d^2)$-root of $\f$.
We prove an auxiliary claim on $\g$.
\begin{claim*}
(a) For any $i\in[d+1]$, 
$g_i(\bx)\geq k_i - 1/(2d)$.

(b) For any subset $I\subseteq [d+1]$,
$\sum_{i\in I} g_i(\bx) > \sum_{i\in I} k_i - 1$.
\end{claim*}
\begin{proof}
(a)
For all $i\in[d]$, the fact that $|\f(\bx)|\leq 1/(2d^2)$ implies that $|g_i(\bx)-k_i|\leq 1/(2d^2)$,
so $g_i(\bx)\geq k_i-1/(2d^2) > k_i - 1/(2d)$.

For $i=d+1$, we have $g_{d+1}(\bx) = n-\sum_{i=1}^d g_i(\bx) \geq n-\sum_{i=1}^d (k_i+1/(2d^2))
= (n-\sum_{i=1}^d k_i) - \sum_{i=1}^d 1/(2d^2)
= k_{d+1}-1/(2d)$.

(b) For every subset $I\subseteq[d+1]$, part (a) implies that
$\sum_{i\in I} g_i(\bx) \geq  \sum_{i\in I} k_i -\sum_{i\in I}1/(2d) =  \sum_{i\in I} k_i - |I|/(2d) > \sum_{i\in I} k_i - 1$, since $|I|\leq d+1<2d$.
\end{proof}
\ifdefined\simplex
Let $\sigma$ be a triangulation simplex that contains $\bx$.  
\fi
\ifdefined\cell
Let $\sigma$ be a grid cell that contains $\bx$.  
\fi
Construct a bipartite graph $G$ in which the $n$ agents are on one side and the $m$ pieces on the other side, and there is an edge between an agent and a piece if the agent prefers the piece in at least one corner of $\sigma$.

For every $i\in[d+1]$, 
$g_i(\bx)$ is a weighted average of $g_i$ on the corners of $\sigma$. 
Hence, claim (a) implies that, 
on at least one corner $\by$ of $\sigma$,
$g_i(\by) \geq k_i - 1/(2d)$ must hold.
Since $g_i$ is an integer on any
grid-point,
$g_i(\by) \geq k_i$ must hold.
Therefore, at least $k_i$ agents must prefer piece $i$ in the partition represented by $\by$, so piece $i$ has at least $k_i$ neighbors in $G$.

Similarly, for every $I\subseteq [d+1]$, 
claim (b) implies that, 
on at least one corner $\by$ of $\sigma$,
$\sum_{i\in I} g_i(\by) \geq \sum_{i\in I} k_i$ must hold.
So the pieces in $I$ have together at least $\sum_{i\in I} k_i$ neighbors in $G$.
Therefore, by Hall's marriage theorem, there is an assignment of agents to pieces such that $k_i$ agents are assigned to piece $i$, which is their neighbor in $G$. 
For each of these agents, if the partition lines are moved by at most $r$ (to a corner of $\sigma$), the agent prefers piece $i$. Therefore, the partition associated with $\bx$ is $r$-near envy-free.

\paragraph{Property \ref{prop:lipschitz}.}
For every $i\in[d]$, $f_i$ is bounded between $-k_i$ and $n-k_i$, and changes linearly between
\ifdefined\simplex
triangulation vertices. 
\fi
\ifdefined\cell
grid-points.
\alex{is this still completely true in the cell approach?}
\erel{I think so, since $\f$ defined as a linear combination of grid points.}
\fi
Therefore, the largest possible change of $f_i$ is a change of $n$ over a distance of $r$. Therefore, $\f$ satisfies the Lipschitz condition with constant $L = n/r$.

\paragraph{Property \ref{prop:switching}.}
For all $i\in[d]$, $x_i=0$ implies that piece $i$ is empty. 
The monotonicity of agents' preferences and the tie-breaking rule ensure that no agent prefers an empty piece. Therefore, $x_i=0$ implies $g_i(\bx)=0$, which implies:
\begin{align*}
    f_i(\bx) = -k_i < 0.
\end{align*}

For all $i\in[d]$, $x_i=1$ implies that pieces $i+1,\ldots,d+1$ are all empty, so $g_{i+1}(\bx)=\cdots=g_{d+1}(\bx)=0$,
so 
$g_{1}(\bx)+\cdots+g_{i}(\bx)=n$.
Therefore,
\begin{align*}
f_{1}(\bx)+\cdots+f_{i}(\bx)=n-(k_1+\cdots+k_i) \geq 0.   
\end{align*}
The above two inequalities imply that $\f$ is sum-switching (see \Cref{def:sum-switching}).

\paragraph{Property \ref{prop:monotone}.}
Fix some $i\in[d]$ and some grid-point $\bx\in[0,1]^d$.
Let $\bx'$ be a grid-point in which, relative to $\bx$, $x_i$ increases by $r$, while $x_j$ for $j\neq i$ remain the same.
The corresponding partition changes as follows:
\begin{itemize}
\item If $x_i\leq (\max_{k \leq i-1} x_k) - r$, which means that piece $i$ is empty, then all pieces remain the same.
\item If $(\max_{k \leq i-1} x_k) \leq x_i$ and $x_i \leq x_{i+1}-r$, then piece $i$ grows and piece $i+1$ shrinks, and all other pieces remain the same.
\item If $(\max_{k \leq i-1} x_k) \leq x_i$ and $x_{i+1} \leq x_i$, then  piece $i$ grows, piece $i+1$ remains empty, some other piece $j>i$ shrinks, and all other pieces remain the same.
\end{itemize}

In all three cases, the monotonicity of agents' preferences (and the consistent tie-breaking rule) implies that any agent who prefers piece $i$ in $\bx$ must prefer piece $i$ in $\bx'$, and any agent who does not prefer piece $i+1$ in $\bx$ must not prefer piece $i+1$ in $\bx'$. Therefore, $g_i(\bx')\geq g_i(\bx)$ and $g_{i+1}(\bx')\leq g_{i+1}(\bx)$, and the same is clearly true for $f_i$ and $f_{i+1}$. 
Therefore, the alternating-monotonicity conditions holds for grid points.

We now prove that the same monotonicity conditions hold for interior points. 
Fix some $i\in[d]$ and some point $\bx\in[0,1]^d$,
and let $\bx'$ be a point in which, relative to $\bx$, $x_i$ increases for some $i\in[d]$, and $x_j$ for $j\neq i$ remain the same.


\ifdefined\simplex
We can decompose the change from $\bx$ to $\bx'$ into small changes that all remain within the same simplex. So without loss of generality, we can assume that $\bx$ and $\bx'$ are in the same simplex; denote this simplex by $\sigma$ and denote its  $d+1$ corners by $\by^0,\ldots,\by^d$.
Let $\pi$ be the permutation used for constructing 
$\sigma$ in the standard triangulation. 
So for all $j\in[d]$, the difference between  $\by^{j-1}$ and  $\by^{j}$ is $r$ in coordinate $\pi(j)$ and $0$ in the other coordinates.


Denote the barycentric coordinates of $\bx$, corresponding to the corners $\by^0,\ldots,\by^d$, by $t_0,\ldots,t_d$, and those of $\bx'$ by 
$t_0',\ldots,t_d'$.
Let $i' := \pi^{-1}(i)$.
As $\bx'-\bx$ is positive in (Cartesian) coordinate $i$ and zero in all other coordinates, 
the differences in the barycentric coordinates are  $t'_{i'}-t_{i'} = \Delta t > 0$ and $t'_{i'-1} - t_{i'-1} = -\Delta t < 0$ and $t'_j = t_j$ for all $j\not \in \{i',i'-1\}$.
Therefore, 
\begin{align*}
\f(\bx') &= \sum_{j=0}^d t'_j \f(\by^j)
\\
         &= \sum_{j=0}^d t_j \f(\by^j)
         + (t'_{i'}-t_{i'})\cdot \f(\by^{i'})
         + (t'_{i'-1}-t_{i'-1})\cdot \f(\by^{i'-1})
\\
         &= \f(\bx)
         + \Delta t \cdot \f(\by^{i'})
         + (-\Delta t) \cdot \f(\by^{i'-1})       
\\
         &= \f(\bx)
         + \Delta t \cdot (\f(\by^{i'}) - \f(\by^{i'-1})).
\end{align*}
Therefore, the change in $\f$ when $\bx$ increases in the $i$-th coordinate is exactly in the same direction as the change in $\f$ when moving along the grid from $\by^{i'-1}$ to $\by^{i'}$, in the $i$-th coordinate. Therefore, $\f$ satisfies the same monotonicity conditions on interior points as on grid points.
\fi
\ifdefined\cell
We can decompose the change from $\bx$ to $\bx'$ into small changes that all remain with the same cell. So without loss of generality, we can assume that $\bx$ and $\bx'$ are in the same cell. 
Denote by $\bx^-$ the point to the ``left'' of $\bx$ and $\bx'$ on their cell boundary, that is, the point $\bx$ in which $x_i$ is replaced by $x_i^-$; similarly, denote by $\bx^+$ the point to the ``right'' of $\bx$ and $\bx'$ on their cell boundary, that is, the point $\bx$ in which $x_i$ is replaced by $x_i^+$. 
By definition, $\f(\bx^-)$ is a weighted average of $2^{d-1}$ grid points with  $i$-th coordinate equal to $x_i^-$, 
whereas $\f(\bx^+)$ is a weighted average of $2^{d-1}$ grid points with $i$-th coordinate equal to $x_i^+$, with exactly the same weights.
Therefore, all monotonicity conditions that hold for grid-points, hold also for $\f(\bx^-)$ and $\f(\bx^+)$. In particular, 
$f_i(\bx^+)\geq f_i(\bx^-)$ and $f_{i+1}(\bx^+)\leq f_{i+1}(\bx^-)$. 

Denote $t_i := (x_i- x_i^-)/r$ and $t'_i := (x'_i- x_i^-)/r$. 
By definition, $\f(\bx) = t_i \f(\bx^+) + (1-t_i)\f(\bx^-)$
and $\f(\bx) = t'_i \f(\bx^+) + (1-t'_i)\f(\bx^-)$.
Since $t_i' > t_i$, the same monotonicity conditions hold for $\bx$ and $\bx'$. In particular,
$f_i(\bx')\geq f_i(\bx)$ and $f_{i+1}(\bx')\leq f_{i+1}(\bx)$. 
\fi
\end{proof}

Combining \Cref{thm:ef-to-root} (for $d=2$) 
with \Cref{thm:root-2d-sum} implies:
\begin{corollary}
\label{thm:near-envy-free}
Given $n$ agents with monotone preferences,
three group-sizes $k_1,k_2,k_3$ with $k_1+k_2+k_3=n$,
and approximation factor $r\in(0,1)$, 
it is possible to find an $r$-near-envy-free cake allocation among three groups using $O(n \log^2(n/r))$ queries.
\end{corollary}

\begin{remark}
In order to extend \Cref{thm:near-envy-free} to $4$ or more groups, 
we would need to use \Cref{thm:dd-sufficient} for $d\geq 3$, but for this we would need all $d^2-d$ ex-diagonal conditions. The current representation of cake-allocations does not guarantee these conditions. 
When $x_i$ grows, every piece $j<i$ does not change. However, this does not imply monotonicity of $g_j$.
As an example, consider an agent who values the piece between $x_{i-1}$ and $x_{i+1}$ more than piece $j$. As $x_i$ moves from $x_{i-1}$ towards $x_{i+1}$, initially the agent prefers piece $i+1$ to piece $j$, so $j$ is not his best piece. As piece $i+1$ shrinks, eventually the agent prefers piece $j$ to piece $i+1$, so piece $j$ \emph{may} become his best piece. As piece $i$ grows, eventually the agent prefers piece $i$ to piece $j$, so piece $j$ is again not the best piece, which means that $g_j$ is not monotone with $x_i$.
Hence, the case of $m\geq 4$ groups remains open.

\end{remark}

\section{Conclusion}
Finding an approximate root is a fundamental question in numeric analysis. This paper investigated the dependence of its run-time on monotonicity conditions.
Besides the three open questions mentioned in the introduction, our results can potentially be extended in several ways:
\begin{itemize}

\item 
Allow the number of functions to differ from the number of variables, that is, 
allow $\f$ to be a function from $\mathbb{R}^{e}$
to $\mathbb{R}^{d}$, where $e\neq d$.
When there are more variables than functions ($e>d$), our positive results still hold, as we can add $e-d$ functions that equal zero identically.
When there are more functions than variables ($e<d$), a root is not guaranteed to exist, but we can investigate the run-time complexity of finding an $\epsilon$-root for functions that have a root.
\item Check whether some monotonicity conditions can be relaxed to \emph{unimodality} (a function is unimodal if it is monotonically-increasing up to some point and then monotonically-decreasing, or vice-versa).

\item Investigate the run-time complexity of related problems, such as finding local minima or maxima, in the presence of monotonicity / unimodality conditions.

\end{itemize}

\section*{Acknowledgments}
We are grateful to Paul Goldberg for the helpful references.
Erel is supported by Israel Science Foundation grant no. 712/20.

\newpage

\bibliographystyle{apalike}
\bibliography{roots}

\end{document}